\documentclass[10pt,aps,pre,twocolumn,groupedaddress]{revtex4-1}
\usepackage{amsmath,amssymb,amsthm,mathrsfs,amsopn}
\usepackage{graphicx}
\usepackage{subfigure}
\usepackage{floatflt}
\newtheorem*{theorem}{Theorem}
\newtheorem{prop}{Proposition}

\begin{document}

\preprint{preprint1}

\title{Analytical approach to network inference: Investigating degree distribution}
\author{Gloria Cecchini}
\thanks{Corresponding author}
\email{gloria.cecchini@abdn.ac.uk}
\affiliation{Institute for Complex Systems and Mathematical Biology, University of Aberdeen, Meston Building, Meston Walk, Aberdeen, AB24 3UE, United Kingdom}
\affiliation{Institute of Physics and Astronomy, University of Potsdam, Campus Golm, Karl-Liebknecht-Stra{\ss}e 24/25, 14476, Potsdam-Golm, Germany}
\author{Bj\"orn Schelter}
\affiliation{Institute for Complex Systems and Mathematical Biology, University of Aberdeen, Meston Building, Meston Walk, Aberdeen, AB24 3UE, United Kingdom}

\date{\today}

\begin{abstract}
When the network is reconstructed, two types of errors can occur: false positive and false negative errors about the presence or absence of links.
In this paper, the influence of these two errors on the vertex degree distribution is analytically analysed.
Moreover, an analytic formula of the density of the biased vertex degree distribution is found.
In the inverse problem, we find a reliable procedure to reconstruct analytically the density of the vertex degree distribution of any network based on the inferred network and estimates for the false positive and false negative errors based on, e.g., simulation studies.
\end{abstract}

\maketitle

\section{Introduction}
 
Networks are one of the most frequently used modelling paradigms for dynamical systems. 
Investigations towards synchronization phenomena in networks of coupled oscillators have attracted considerable attention, and so has the analysis of chaotic behaviour and corresponding phenomena in networks of dynamical systems to name just a few \cite{Arkady,Arkady2016,Rok2017,Li2014}. 
Understanding and characterizing network behaviour has triggered interest in a vast number of disciplines, ranging from optimizing vaccination strategies \cite{Clusella2016} to understanding the functioning or malfunctioning of the human brain \cite{Bullmore2009,Pessoa2014, Petersen2015}.

While first principle modelling is feasible in some areas, in others, networks need to be inferred, e.g. from observed data, see, e.g., \cite{Guimera2009,Arkady2011,Karrer2011,Alderisio2017}. 
This comes with certain challenges ranging from selecting the appropriate nodes or even defining them, to the choice of an appropriate technique to infer the interaction between the nodes.
These choices have a strong impact on the resulting network. 
Here, we discuss another related challenge that originates from the fact that network inference in the \textit{Inverse Problem} typically relies on statistical methods and selection criteria.
 
A typical network inference procedure, estimates the connectivity between a-priori specified nodes in a network. 
If the connectivity measure passes a certain threshold, a link between the corresponding nodes is assumed to be present. 
The choice of this threshold is arbitrary, but it is intuitively clear that there is a strong correlation between number of links and choice of threshold. 
Selecting the threshold, not only controls how many links are inferred correctly but also establishes the number of incorrectly determined links. 
There are two types of errors,
(i) a link may be erroneously considered present, this false positive conclusion is referred to as a \textit{type I error};
(ii) a present link may remain undetected, this false negative conclusion is referred to as a \textit{type II error}.

In this manuscript, we present an analytical framework that on the network level links the reconstructed network structure contaminated by \textit{type I} and \textit{type II errors} to the true underlying one. 
While the framework is rather general, we used the vertex degree distribution to derive the functional relationship between the reconstructed and true underlying network. 
This enables us to obtain superior estimates for the vertex degree distribution, a key property of a network \cite{NewmanBook}.
It has been shown that including the vertex degrees into stochastic blockmodels improves their performance for statistical inference of group structure \cite{Karrer2011}.
The functional relationship depends on the choice of \textit{type I error}, \textit{type II error} and the dimension of the network.
We demonstrate the performance of our novel approach in a simulation study.
 
The manuscript is structured as follows. 
In Section \ref{Materials and Methods} a theoretical analysis of our method is presented.
Section \ref{results} shows some cases where the method presented in Section \ref{Materials and Methods} is applied.

\section{Materials and Methods}
\label{Materials and Methods}

In section \ref{networkChange}, a short introduction to networks is presented; we analyse the influence of \textit{type I} and \textit{type II errors} on the network structure, i.e. false positive and false negative conclusions about links.
In section \ref{inferenceVertDegree} different methods to solve the \textit{Inverse Problem} are presented.
Section \ref{dirNetworks} contains a brief description of the generalization to directed networks.

\subsection{Networks change}
\label{networkChange}

A network is defined as a set of vertices (or nodes) with links (or edges) between them. 
To quantify the structure of networks, different characteristics have been introduced \cite{Olbrich2010}. 
Here, we consider two key network characteristics: vertex degree distribution and number of edges.
The vertex degree describes the number of links of a node; if the vertex $v$ has $k$ edges attached, its vertex degree is $d=k$.
The vertex degree distribution is an important property of the entire network.

Networks can be directed or undirected \cite{NewmanBook}.
In an undirected network, connection of $v_1$ to $v_2$ implies the connection of $v_2$ to $v_1$.
In a directed network, this symmetry is broken, therefore if a path from $v_1$ to $v_2$ exists, a path from $v_2$ to $v_1$ does not necessarily exist.
In this section, we consider undirected networks.
Later [Sec.~\ref{dirNetworks}] a generalization to directed networks is presented.

Consider a network $G$ with $n$ nodes and vertex degree distribution defined by the probability function  $\mathcal{P}$, i.e., $\mathcal{P}_i=\mathbb{P}(d=i)$ is the probability that the degree $d$ is $i$, for $i=0,\cdots,n-1$.
Note that the degree of a vertex is between $0$ and $n-1$, since each vertex can be connected to at most $n-1$ remaining vertices.

We analyse the influence of \textit{type I} and \textit{type II errors} on the vertex degree distribution of a given network $G$.
We call $G'$ the network detected when \textit{type I} and \textit{type II errors} occur and we assume that $\alpha$ is the probability of a \textit{type I error} and that $\beta$ is the probability of a \textit{type II error}.
Therefore, $\alpha$ expresses the probability that a link absent in $G$ is present in $G'$ and $\beta$ is the probability that a link present in $G$ is no longer present in $G'$.
Hence, the set of edges of $G'$ is a combination of \textit{true positive links} and \textit{false positive links} of $G$.
The vertex degree distribution of $G'$ is characterised by the probability function $\mathcal{P'}$.

Consider a vertex and assume it has degree $k$, therefore there are $k$ links connected to it and $n-1-k$ absent links.
We evaluate the probability that this vertex has vertex degree $k'$ in $G'$.
The vertex degree
\begin{equation}
\label{sum}
k'=j+i
\end{equation}
is given by the sum of \textit{true positive links} $j$ and \textit{false positive links} $i$; additionally, $i$ and $j$ have to satisfy
\begin{subequations}\label{cond}
\begin{align} 
& j\leq k\qquad \text{and}\label{cond1}\\
& i\leq n-1-k.\label{cond2}
\end{align}
\end{subequations}
The condition described by Eq.~(\ref{cond1}) guarantees that the number of false negative links is larger or equal than zero, and smaller or equal than the number of the original true positive links, i.e., $0\leq k-j\leq k$.
Likewise, the number of false positive links must be non-negative and smaller or equal than the number of the original non-present links, Eq.~(\ref{cond2}).

The probability that a vertex has degree $k'$ in $G'$, knowing it has degree $k$ in $G$ is

\vspace{0.5cm}

\begin{widetext}
\begin{equation}
\mathbb{P}(d'=k'|d=k)=
\begin{cases}
{\displaystyle \sum_{i=0}^{k'}\binom{k}{k'-i}(1-\beta)^{k'-i}\beta^{k-k'+i}\binom{n-1-k}{i} \alpha^{i}(1-\alpha)^{n-1-k-i}}\qquad\qquad\qquad\qquad &\text{ }\\
&\phantom{\text{ }}\llap{\text{if  $k'\leq k$ and $k'\leq n-1-k$}}\\
{\displaystyle \sum_{i=0}^{k}\binom{k}{i}(1-\beta)^{i}\beta^{k-i}\binom{n-1-k}{k'-i} \alpha^{k'-i}(1-\alpha)^{n-1-k-k'+i}}\\
&\phantom{\text{ }}\llap{\text{if  $k<k'\leq n-1-k$}}\\
{\displaystyle \sum_{i=0}^{n-1-k'}\binom{k}{k-i}(1-\beta)^{k-i}\beta^{i}\binom{n-1-k}{k'-k+i} \alpha^{k'-k+i}(1-\alpha)^{n-1-k'-i}}\\
&\phantom{\text{ }}\llap{\text{if  $k'\geq k$ and $k'>n-1-k$}}\\
{\displaystyle \sum_{i=0}^{n-1-k}\binom{k}{k'-i}(1-\beta)^{k'-i}\beta^{k-k'+i}\binom{n-1-k}{i} \alpha^{i}(1-\alpha)^{n-1-k-i}}\\
&\phantom{\text{ }}\llap{\text{if  $n-1-k<k'<k$.}}
\end{cases}
\label{bigEq}
\end{equation}
\end{widetext}

The probability $\mathbb{P}(d'=k'|d=k)$ is a piecewise function for all combinations of $i$ and $j$ satisfying Eqs.~(\ref{sum}) and (\ref{cond}).
To obtain Eq.~(\ref{bigEq}) we consider, as an example, the first case, i.e., $k'\leq k$ and $k'\leq n-1-k$.

The probability of having $j$ \textit{true positive links}, over all possible $k$ original true positive links, is 
\begin{equation}
\label{falseNeg}
\mathbb{P}(\text{\# true positive links}=j)=\binom{k}{j}(1-\beta)^j \beta^{k-j},
\end{equation}
which is the binomial distribution $\mathcal{B}(k,1-\beta)$.
Since $j=k'-i$ [Eq.~(\ref{sum})], Eq.~\ref{falseNeg} corresponds to the first part of the first case of Eq.~\ref{bigEq}.
Similarly, the probability of having $i$ \textit{false positive links} is 
\begin{equation}
\label{falsePos}
\mathbb{P}(\text{\# false pos links}=i)=\binom{n-1-k}{i}\alpha^i (1-\alpha)^{n-1-k-i},
\end{equation}
which is the binomial distribution $\mathcal{B}(n-1-k,\alpha)$. 

Combining Eqs.~(\ref{falseNeg}) and (\ref{falsePos}), changing variable $j$ according to Eq.~(\ref{sum}), and considering all possible combinations of $i$ and $j$, we obtain the first case of Eq.~(\ref{bigEq}).
All the other cases can be derived in the same way following the conditions in Eq.~(\ref{cond}).

Applying the law of total probability

\begin{multline}
\label{totProb}
\mathbb{P}(d'=k')=\sum_{k=0}^{n-1}\mathbb{P}(d'=k'|d=k) \mathbb{P}(d=k) \\ \forall k'\in \{0,\cdots,n-1\}
\end{multline}
we obtain the matrix equation

\begin{multline*}
\left[
\begin{smallmatrix}
\mathbb{P}(d'=0)\\
\vdots\\
\mathbb{P}(d'=n-1)
\end{smallmatrix}
\right]
= \\
\underbrace{\left[
\begin{smallmatrix}
\mathbb{P}(d'=0|d=0) & \dotsb & \mathbb{P}(d'=0|d=n-1)\\
\vdots&\ddots&\vdots\\
\mathbb{P}(d'=n-1|d=0)&\cdots&\mathbb{P}(d'=n-1|d=n-1)\\
\end{smallmatrix}
\right]}_{=A}
\cdot
\underbrace{\left[
\begin{smallmatrix}
\mathbb{P}(d=0)\\
\vdots\\
\mathbb{P}(d=n-1)
\end{smallmatrix}
\right]}_{=\mathcal{P}},
\end{multline*}

\noindent
i.e.,

\begin{equation}
\label{p1Ap}
\mathcal{P'}=A\mathcal{P}.
\end{equation}

The matrix $A=A(n,\alpha,\beta)$ depends on $n$, $\alpha$ and $\beta$ and has determinant 
\begin{equation}
\label{detA}
\det A=(1-\alpha-\beta)^{\frac{n(n-1)}{2}}, 
\end{equation}
therefore it is invertible if and only if $\alpha\neq 1-\beta$, see Appendix \ref{apProof} for a proof.

Assuming $G$ is known, Eq.~(\ref{p1Ap}) characterises the influenced of \textit{type I} and \textit{type II errors} on the vertex degree distribution, and it allows to find the vertex degree distribution of the network $G'$.

\subsection{Inference of networks' vertex degree distribution}
\label{inferenceVertDegree}

Section \ref{networkChange} analyses the impact of \textit{type I} and \textit{type II errors} on the vertex degree distribution of a given network.
Equation (\ref{p1Ap}) allows to obtain $\mathcal{P'}$ from $\mathcal{P}$.
In this section, we are interested in the inverse problem, i.e., inverting Eq.~(\ref{p1Ap}), to infer the original vertex degree distribution from an observed one.
When $\{\alpha,\beta\}\neq \{0,0\},\{1,1\}$, since the convergence to zero of the determinant of $A$ scales like $x^{\frac{n(n-1)}{2}}$ for $|x|<1$ [Eq.~(\ref{detA})], numerical issues arise for relatively small $n$ when inverting the matrix $A$ to find $\mathcal{P}$ through $\mathcal{P}=A^{-1}\mathcal{P'}$.
The cases $\alpha,\beta=0$ and $\alpha,\beta=1$ are trivial, see Appendix \ref{apProof}.

The \textit{least squares method} is a standard approach to solve problems like Eq.~(\ref{p1Ap}).
Although the matrix $A$ is not singular, for reasonable parameter values for $n$, $A$ is typically ill-conditioned, therefore the pseudoinverse of the truncated singular value decomposition of $A$ is used.

The \textit{singular value decomposition} of a matrix $A$ is the factorization of the matrix into the product of $A=UWV^T$ where $W$ is a diagonal matrix and the columns of the matrices $U$ and $V$ are orthonormal \citep{Yanai2011}.
The elements $w_1,\cdots,w_n$ on the diagonal of $W$ are called \textit{singular values} of $A$ and they are ordered such that $w_1\geq w_2\geq\cdots\geq w_r>w_{r+1}=\cdots =w_n=0$, where $r$ is the rank of $A$.

The singular value decomposition is a tool to compute the pseudoinverse of a matrix.
If $A$ has singular value decomposition $A=UWV^T$, its pseudoinverse $A^+$ is defined as $A^+=VW^+U^T$, where $W^+$ is obtained from $W$ replacing all the non-zero elements with their reciprocals.

The \textit{truncated singular value decomposition} is a method for regularization of ill-posed least squares problems \citep{Hansen1987}.
Once the singular value decomposition $A=UWV^T$ is found, the matrix $W$ is truncated at, e.g., rank $t$ such that only the first $t$ \textit{singular values} are considered; this matrix is usually called $W_t$.
More precisely, $W_t$ is a diagonal matrix with elements $w_1\geq w_2\geq\cdots\geq w_t>w_{t+1}=\cdots =w_n=0$, with $t<r$.
The truncated diagonal matrix $W_t$ is used to find an approximation of the matrix $A$ using its decomposition, i.e., $A_t=UW_tV^T$.
The optimal value for $t$ has been studied in \citep{Frank,Gavish2014}.
The matrix $A_t$ is the closest approximation of $A$ of rank $t$, \citep{Hansen1987}.
Using $W_t$, we calculate the pseudoinverse of $A_t$, i.e., $A_t^+=VW_t^+U^T$, and we solve Eq.~(\ref{p1Ap}) resulting in

\begin{equation}
\label{pA+p1}
\mathcal{P}=A_t^+\ \mathcal{P'}.
\end{equation}

\subsection{Generalization for directed networks}
\label{dirNetworks}

For directed networks the vertex degree is characterised by the \textit{vertex in-degree} and the \textit{vertex out-degree}, \cite{NewmanBook}.
Usually, in a directed network the \textit{vertex degree} is the sum of the vertex in-degree and the vertex out-degree.

Both the in-degree and the out-degree of a vertex are numbers between $0$ and $n-1$, if $n$ is the number of vertices of the network. 
Therefore, the analysis shown in Section \ref{networkChange} and \ref{inferenceVertDegree} remains valid if either the vertex in-degree or the vertex out-degree are considered instead of the vertex degree.

An undirected network with $n$ nodes has at most $n(n-1)/2$ edges.
A network with $n$ nodes has at most $n(n-1)$ directed edges;
a generalization for other characteristics is likely more complicated, and therefore requires a more in-depth analysis.

\section{Results and Discussions}
\label{results}

To demonstrate the abilities as well as limitations, the analysis presented in Section \ref{Materials and Methods} is applied to some typical simulated networks.
We like to highlight that our approach is derived analytically; simulation studies are predominantly needed to demonstrate its applicability in real-world examples and to check for numerical issues, etc.
There might be practical issues, e.g., due to the dimension of the network, and with the aim to show how these challenges can be overcome we present a simulation study to explore the concrete applicability of our method. 

We study 5 network topologies that present different characteristics so to have a spectrum of networks as wide as possible to which we apply our analysis.
Namely, we consider Erd{\H o}s-R{\'e}nyi (also called random), Small-World, Scale-Free networks, a three-dimensional grid, and a network of randomly connected communities, \cite{NewmanBook}.
We vary the probabilities $\alpha$ and $\beta$ of \textit{type I} and \textit{type II errors} in the range $1\%-10\%$ mimicking a typical analysis method that has high sensitivity and high specificity.
Nevertheless, both lower and higher values for $\alpha$ and $\beta$ can be chosen and the results obtained are qualitatively the same as the ones presented below.

Consider a random network $G$ with $100$ nodes and a probability of a connection of $0.2$.
The vertex degree has binomial distribution $\mathcal{B}(100,0.2)$.
Adding and removing links with probabilities $\alpha=0.05$ and $\beta=0.03$ results in a new network $G'$.
The vertex degree distribution of $G'$ is calculated empirically by counting the vertices' degrees.
Applying the procedure explained above, the vertex degree distribution of the original network is estimated.
Figure \ref{ErdosAll} shows the results using the cut-off for the truncated singular value decomposition method of $0.5$, i.e., $W_t$ contains only singular values greater than $0.5$.
The choice of $t$ is motivated by smoothness and regularity of the solution obtained.

Figure~\ref{ErdosAll} shows the histogram of the degrees of the vertices of the original network $G$, the density of the detected network $G'$, the reconstructed vertex degree distribution of the original network $\mathcal{P}$ resulting from Eq.~(\ref{pA+p1}), and the result when a non-negative constraint is applied to the truncated singular value decomposition to avoid that numerical issues result in negative solutions.
More precisely, we use $lsqnonlin$ Matlab function with lower bound condition $lb=zeros(n)$; this function implements the \textit{trust region reflective} algorithm \cite{Coleman1994,Branch1999}.
The density of $G'$ is estimated by 
$$\mathcal{P'}_i=\frac{\text{number of nodes with vertex degree}=i}{\text{number of nodes}}$$
its empirical distribution, and this is used to infer the original network.

\begin{figure}[t!]
\includegraphics[width=0.5\textwidth,trim=0.9cm 0cm 1cm 0cm,clip=true]{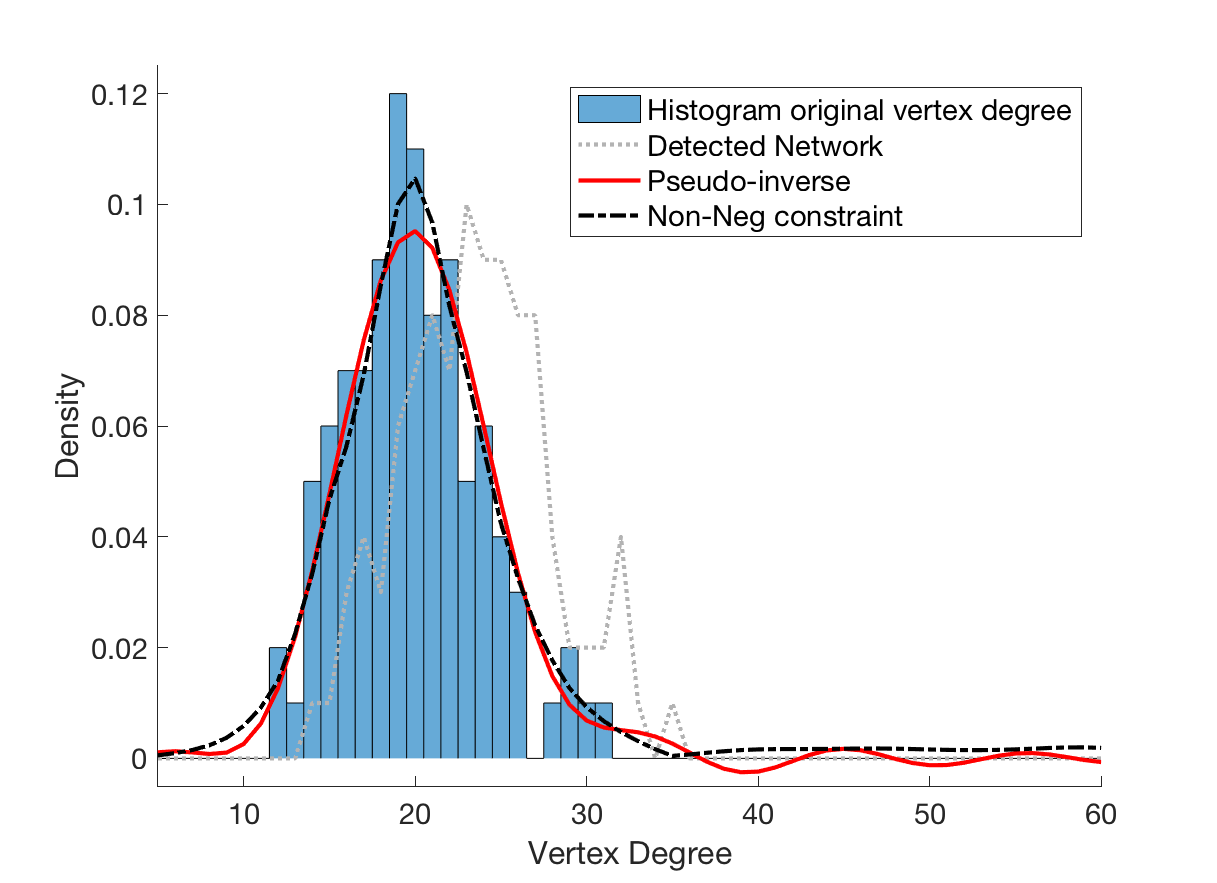}
\caption{Density histogram of the original vertex degrees, blue bars, detected density vertex degree distribution, gray dotted line, the result of network reconstruction using Eq.~(\ref{pA+p1}) knowing $A$ and $\mathcal{P'}$, solid red line, and the result when a non-negative constraint is applied to the truncated singular value decomposition, black dashed line. The original network is a random network with $100$ nodes and probability of connection $0.2$.}
\label{ErdosAll}
\end{figure}

Figure \ref{SmallAll} shows the result when the original network $G$ is a Small-World network.
It is built from the regular network of $100$ nodes, vertex degree $4$, and probability of rewiring $0.4$.
The network $G'$ is obtained by adding and removing links at random with probabilities $\alpha=0.03$ and $\beta=0.05$ respectively.
The cut-off for the truncated singular value decomposition method is $0.33$.

Figure \ref{SmallAll} shows the histogram of the degrees of the vertices of the original network $G$, the density of the detected network $G'$, the reconstructed vertex degree distribution, and the result when a non-negative constraint is applied to the truncated singular value decomposition.

\begin{figure}[t!]
\centering
\includegraphics[width=0.5\textwidth,trim=0.7cm 0cm 1cm 0cm,clip=true]{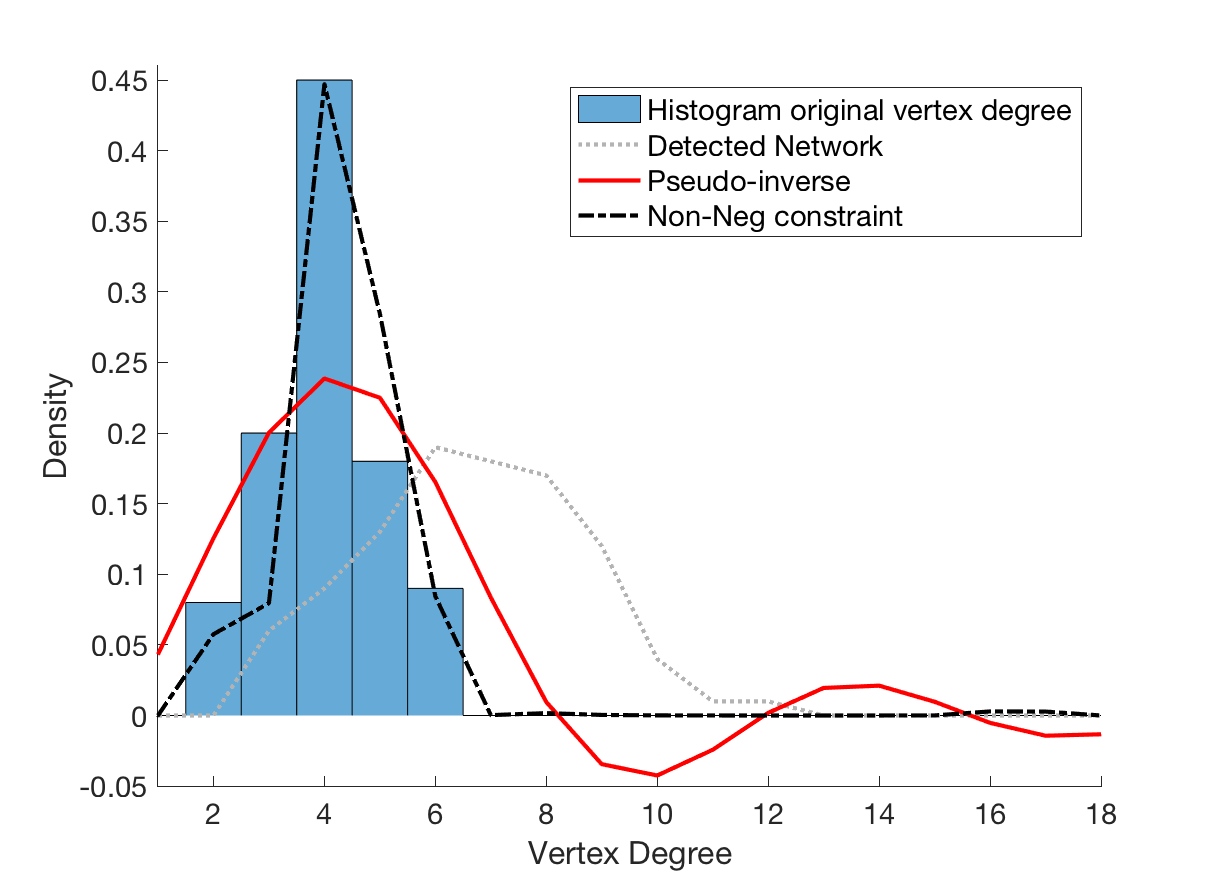}
\caption{Density histogram of the original vertex degrees, blue bars, detected density vertex degree distribution, gray dotted line, the result of network reconstruction using Eq.~(\ref{pA+p1}) knowing $A$ and $\mathcal{P'}$, solid red line, and the result when a non-negative constraint is applied to the truncated singular value decomposition, black dashed line. The original network is a Small World network with $100$ nodes and probability of rewiring $0.4$.}
\label{SmallAll}
\end{figure}

Figure \ref{SFAll} shows the result when the original network $G$ is a Scale-Free network.
It is built using a preferential attachment model for network growth.
At each step a vertex, with a link attached to it, is added.
The probability that the new vertex attaches to a given old one is proportional to its vertex degree.
This procedure is repeated until the network has $100$ nodes.
The network $G'$ is obtained by adding and removing links at random with probabilities $\alpha=0.1$ and $\beta=0.03$ respectively.
The cut-off for the truncated singular value decomposition method is $0.4$.

Figure \ref{SFAll} shows the histogram of the degrees of the vertices of the original network $G$, the density of the detected network $G'$, the solution of Eq.~(\ref{pA+p1}), and the result when a non-negative constraint is applied to the truncated singular value decomposition.

\begin{figure}[t!]
\centering
\includegraphics[width=0.5\textwidth,trim=1cm 0cm 1cm 0cm,clip=true]{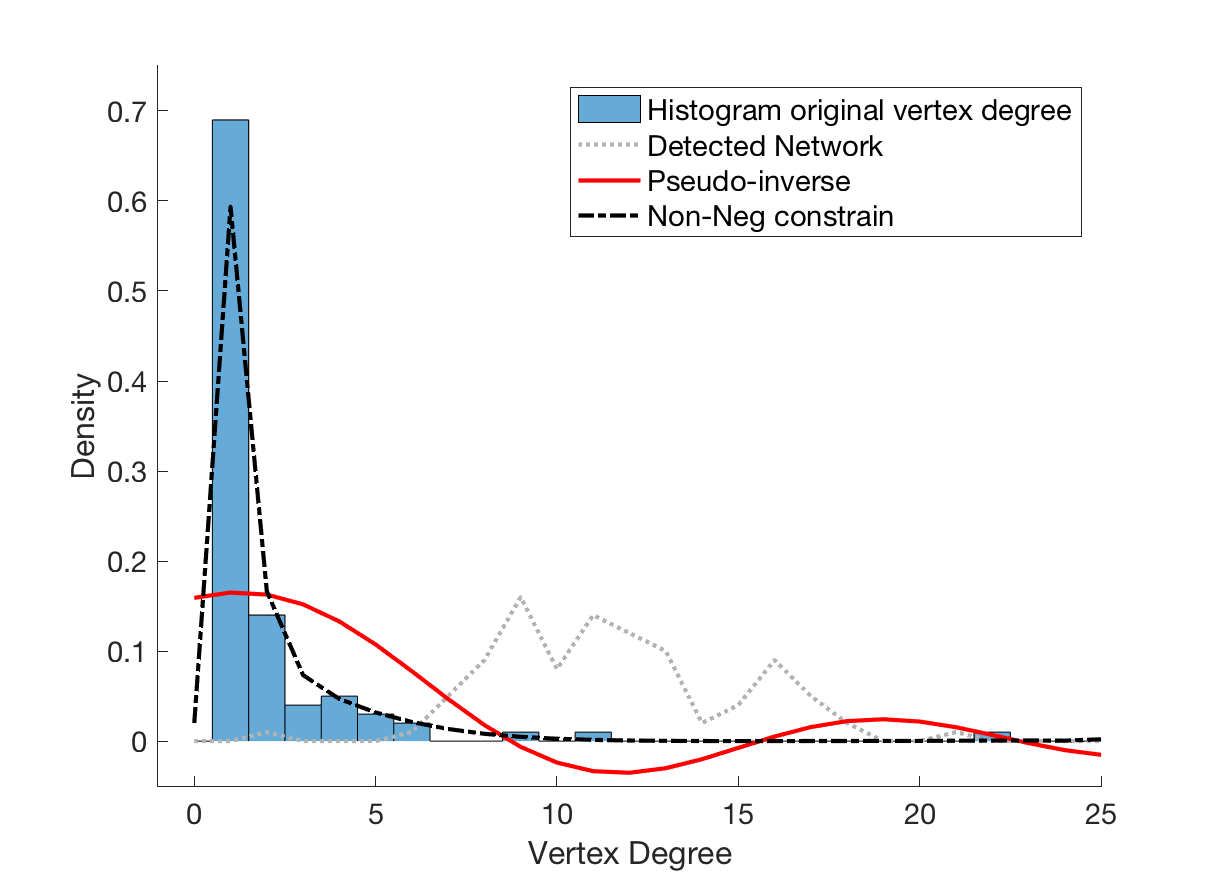}
\caption{Density histogram of the original vertex degrees, blue bars, detected density vertex degree distribution, gray dotted line, the result of network reconstruction using Eq.(\ref{pA+p1}) knowing $A$ and $\mathcal{P'}$, solid red line, and the result when a non-negative constraint is applied to the truncated singular value decomposition, black dashed line. The original network is a Scale-Free network with $100$ nodes.}
\label{SFAll}
\end{figure}

Another example we apply our method to, is when the original network $G$ is a three-dimensional grid $4\times 5\times 5$; note that $G$ has $100$ nodes.
The network $G'$ is obtained by adding and removing links at random with probabilities $\alpha=0.1$ and $\beta=0.05$ respectively.
The cut-off for the truncated singular value decomposition method is $0.38$.
Figure \ref{AssortAll} shows the histogram of the degrees of the vertices of the original network $G$, the density of the detected network $G'$, the solution of Eq.~(\ref{pA+p1}), and the result when a non-negative constraint is applied to the truncated singular value decomposition.

\begin{figure}[t!]
\centering
\includegraphics[width=0.5\textwidth,trim=1cm 0cm 1cm 0cm,clip=true]{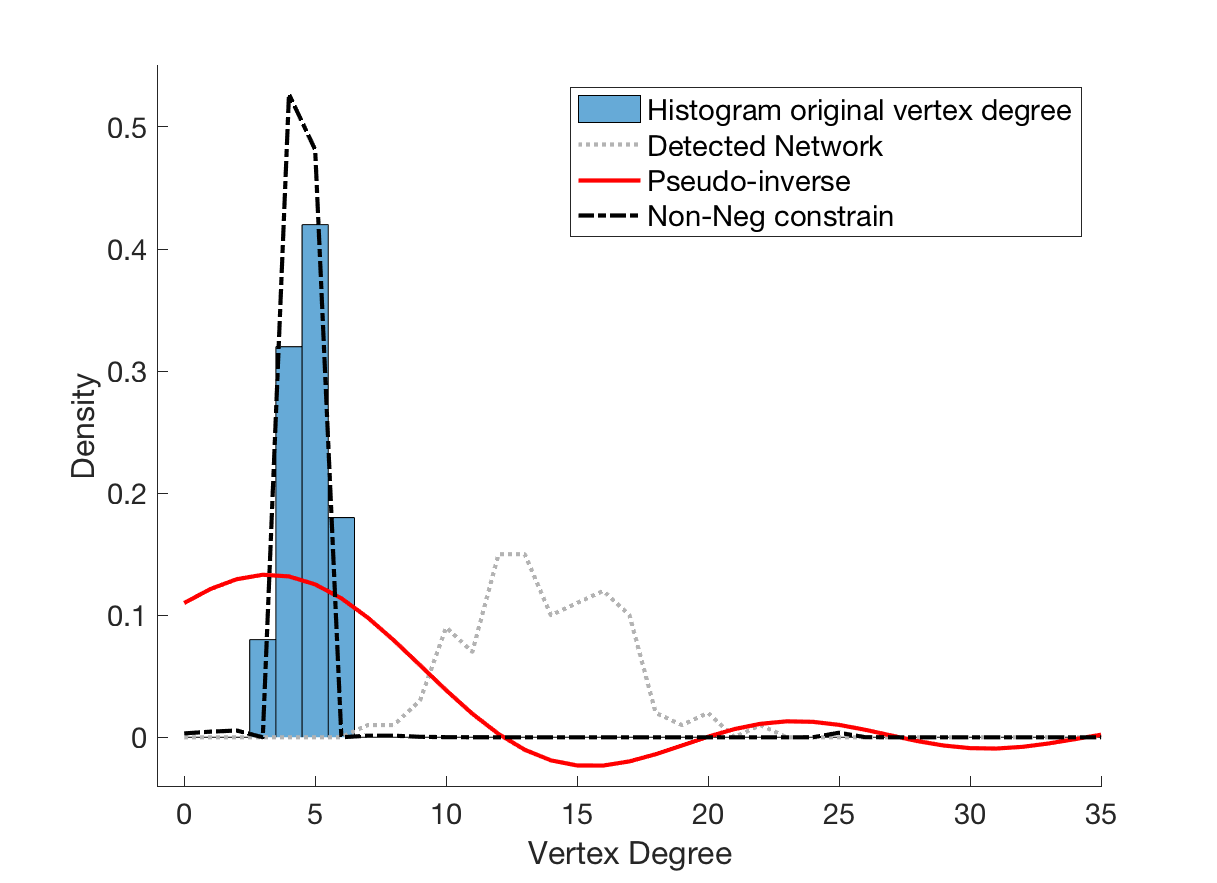}
\caption{Density histogram of the original vertex degrees, blue bars, detected density vertex degree distribution, gray dotted line, the result of network reconstruction using Eq.(\ref{pA+p1}) knowing $A$ and $\mathcal{P'}$, solid red line, and the result when a non-negative constraint is applied to the truncated singular value decomposition, black dashed line. The original network is a $4\times 5\times 5$ grid.}
\label{AssortAll}
\end{figure}

Finally, we investigate the case when $G$ is a network of three randomly connected communities.
It is built by constructing three Erd{\H o}s-R{\'e}nyi networks, with probability of connection $0.3$, $0.6$, and $0.9$, and each with $33$ nodes.
Then, nodes from different communities are connected with probability $0.1$.
The network $G'$ is obtained by adding and removing links at random with probabilities $\alpha=0.05$ and $\beta=0.03$ respectively.
The cut-off for the truncated singular value decomposition method is $0.42$.
Figure \ref{CommAll} shows the histogram of the degrees of the vertices of the original network $G$, the density of the detected network $G'$, the solution of Eq.~(\ref{pA+p1}), and the result when a non-negative constraint is applied to the truncated singular value decomposition.

\begin{figure}[t!]
\centering
\includegraphics[width=0.5\textwidth,trim=0.9cm 0cm 1cm 0cm,clip=true]{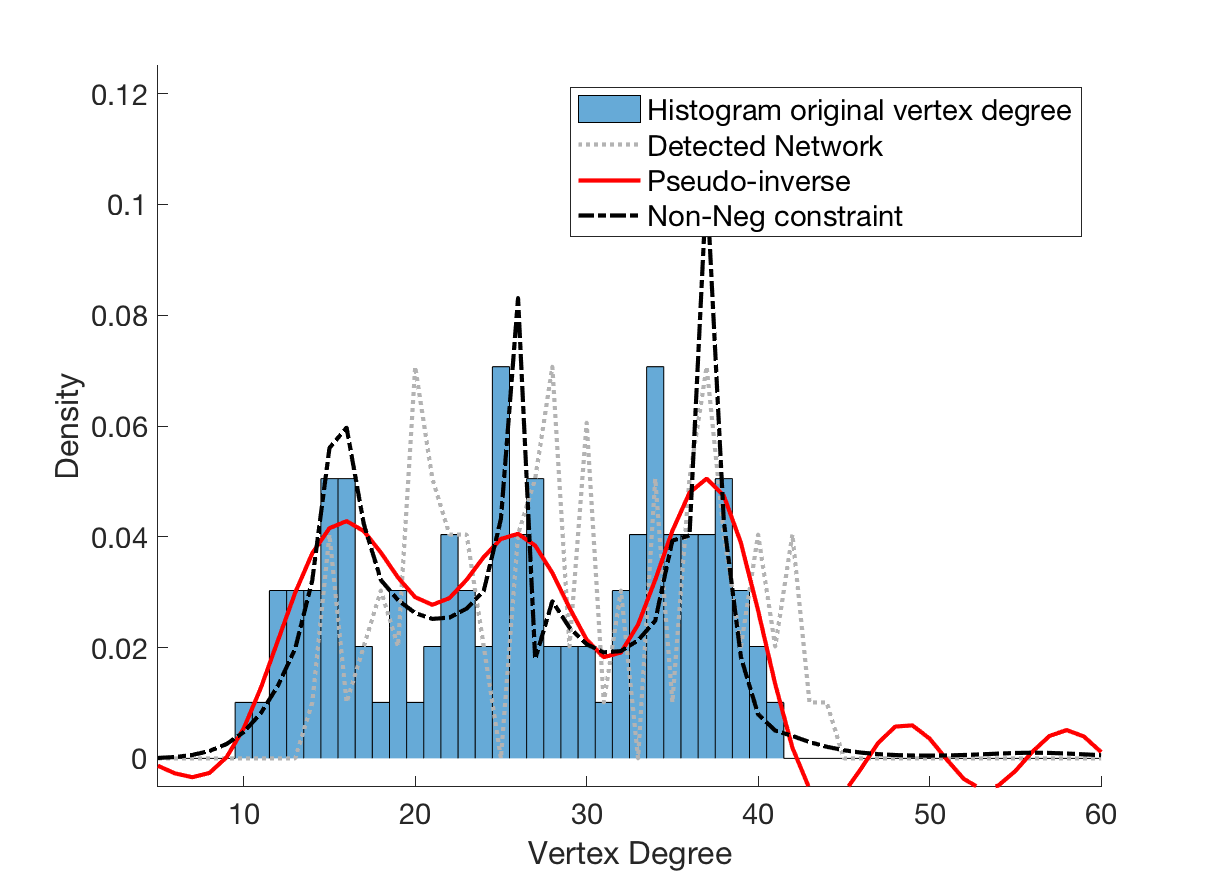}
\caption{Density histogram of the original vertex degrees, blue bars, detected density vertex degree distribution, gray dotted line, the result of network reconstruction using Eq.(\ref{pA+p1}) knowing $A$ and $\mathcal{P'}$, solid red line, and the result when a non-negative constraint is applied to the truncated singular value decomposition, black dashed line. A network of three randomly connected communities is used as original network.}
\label{CommAll}
\end{figure}

Another interesting aspect is the influence of \textit{type I} and \textit{type II errors} and the proposed method on the reconstruction of individual nodes and not just the correct distribution. 
This is particularly relevant for nodes that have a degree much higher than average, so-called hubs.

In the Scale-Free example, Fig.~\ref{SFAll}, the detected distribution appears to be smoother than the original, implying that a hub might have been converted to a non-hub. 
Analysing this in more detail, there is convincing evidence that this is not the case - hubs are correctly identified as hubs.

Consider a node $d$ that has degree $k$ in $G$ that has $n$ nodes.
Due to \textit{type I} and \textit{type II errors}, this node in $G'$ has degree $d'$, a random variable with distribution shown in Eq.~(\ref{bigEq}).
Taking realisations of this random variable, and inverting the process using Eq.~(\ref{pA+p1}), allows us to compare individual degrees for a given node of the true network with the reconstructed one.
We consider a network with $n=100$ nodes, a node $d$ with degree $k=75$, probabilities of \textit{type I} and \textit{type II errors} of $\alpha=0.05$ and $\beta=0.03$, respectively, and we simulate $100$ realisations of the random variable described above.
Figure~\ref{hub1} shows the reconstruction of the degree of $d$ using these realisations.
The result does not only show an improvement from the detected degrees $k’$, but also illustrates the high accuracy of the reconstruction method.

Figure~\ref{hubAllin1} shows the reconstruction of various degrees, i.e., $k$ from $10$ to $90$ in steps of $10$, using the same parameters $n=100$, $\alpha=0.05$, $\beta=0.03$, and $100$ realisations each.
This again demonstrates that our method reliable reconstructs the correct degree for this individual node. 
Further simulations, not presented here, varying $\alpha$ and $\beta$ between $0.01$ and $0.1$, show qualitatively the same results. 
In every case, the reconstruction is very robust, and this suggests that it is extremely unlikely that a hub is reconstructed as a non-hub. 
Moreover, the reconstruction works correctly not only on the general distribution, but also when it is applied to single nodes.

\begin{figure}[!t]
\includegraphics[width=0.5\textwidth,trim=1cm 0cm 1cm 0cm,clip=true]{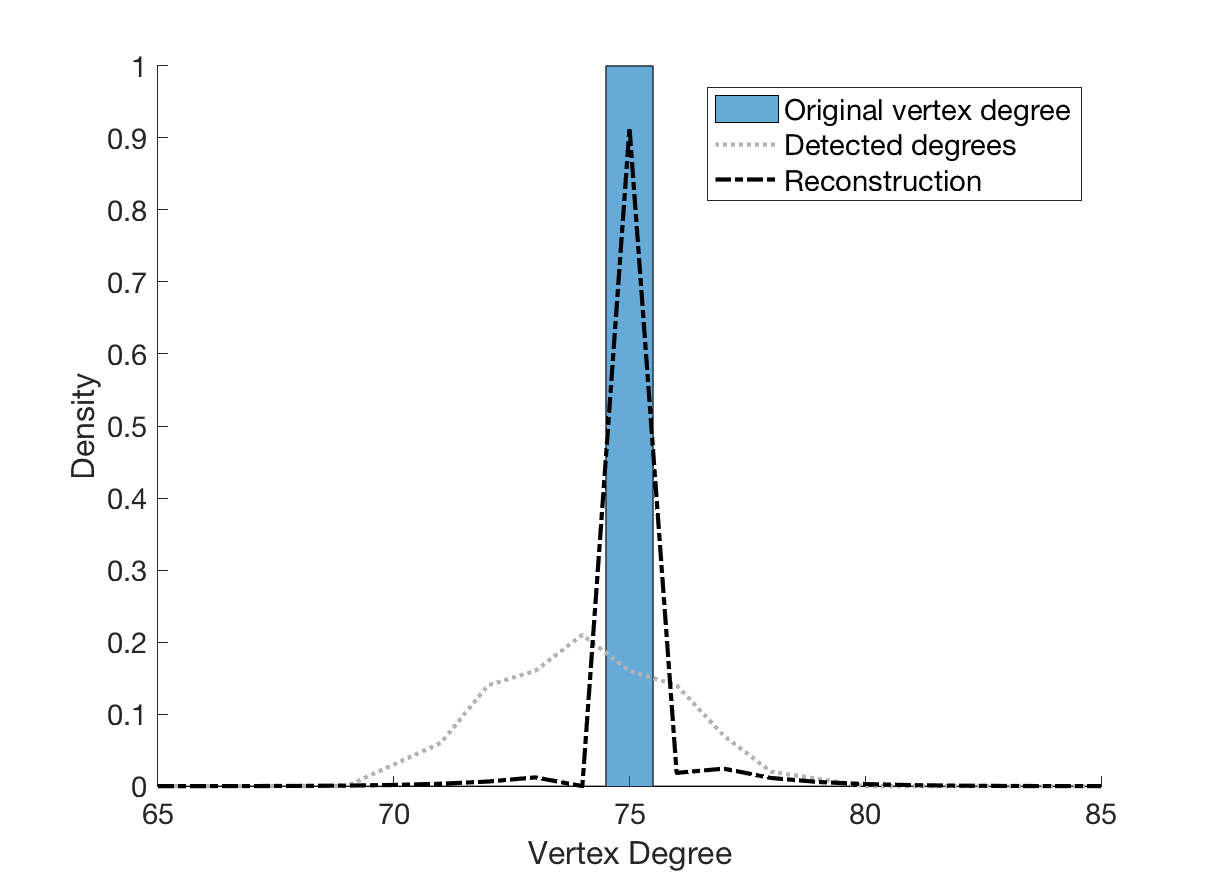}
\caption{Reconstruction of the degree of a single node with original degree $k=75$.}
\label{hub1}
\end{figure}

\begin{figure}[!t]
\includegraphics[width=0.5\textwidth,trim=1cm 0cm 1cm 0cm,clip=true]{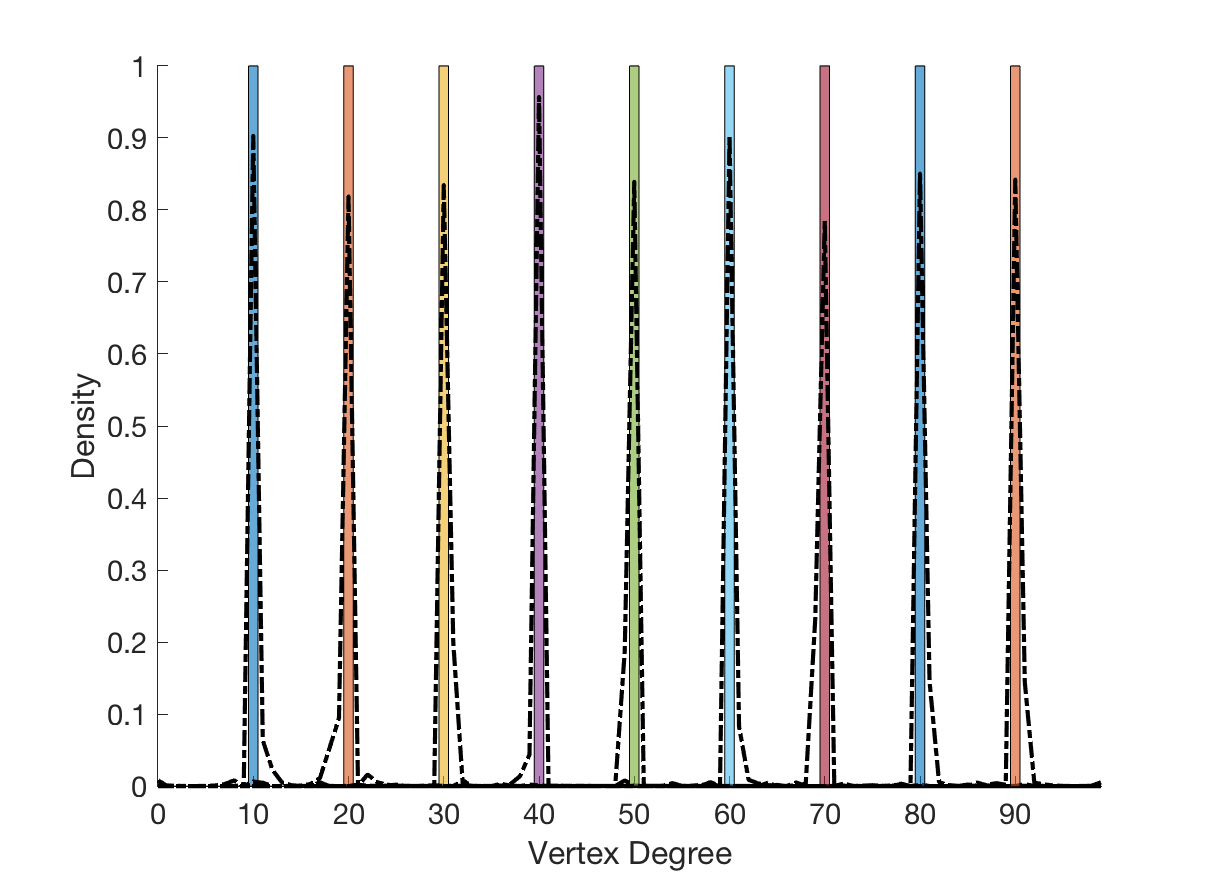}
\caption{Reconstructions of the degree of single nodes with original degrees $k$ from $10$ to $90$ in steps of $10$.}
\label{hubAllin1}
\end{figure}

\subsection{Robustness of reconstruction}

As stated above, our reconstruction method assumes the probabilities of \textit{type I} and \textit{type II errors} to be known a priori.
While the \textit{type I error} is controlled by statistical methods, the \textit{type II error} must be inferred or reasonable assumptions from simulations, or prior studies, about the \textit{type II error} must be available. 
To show the impact of violations of this and thereby the robustness of our method, we analyse the performance of the reconstruction when perturbations on $\alpha$ and $\beta$ are introduced.

Figures \ref{wrongB2}-\ref{wrongA2} demonstrate the robustness of our approach for various examples. 
Figures \ref{wrongB2}-\ref{wrongB} are used to show robustness with respect to $\beta$, while Figs. \ref{wrongA}-\ref{wrongA2} show the robustness with respect to $\alpha$. 
The perturbations are quantified in percentage using the parameter $\delta$, e.g. the perturbations of $\beta$ are expressed by $\beta+\delta\beta$. Note that, since $0\leq\beta\leq 1$, the conditions for the perturbations are $-1\leq\delta\leq 1/\beta-1$; namely, if we call $\beta^p=\beta+\delta\beta$ the perturbed $\beta$, then we have
\begin{equation}
\begin{split}
0\leq&\beta^p\leq 1\\
0\leq&\beta+\delta\beta\leq 1\\
-1\leq&\delta\leq 1/\beta-1.
\end{split}
\end{equation}
Negative values for $\delta$ represent underestimated values for $\beta$ and positive overestimated values for $\beta$.
The same argument is used for the perturbation of $\alpha$.

Figure~\ref{wrongB2} shows the reconstruction of a Scale-Free network with $100$ nodes for the true value of $\beta=0.03$, and also for various values of $\beta$ deviating up to $1000\%$ from the true value, i.e., $\beta^p=0.33$.
The cut-off for the truncated singular value decomposition method is $0.1$ and the probability of \textit{type I error} is $\alpha=0.05$, assuming to control the family-wise error rate at this value, i.e., the probability of making at least one \textit{type I error}; it is beyond the scope of this manuscript to discuss cases in which the technique selected to reconstruct the network violates this assumption - we will however estimate the results for different deviations from the true $\alpha$ used to generate the plots to investigate its robustness. 
Figure~\ref{wrongB2} shows that our approach is robust to rather large perturbations of $\beta$, in both negative and positive directions. 
Up to $\delta=500\%$, the bias of the reconstruction is negligible; only if $\delta=1000\%$ or more deviates the reconstruction significantly from the true one, although it still performs better than the na\"ive approach of trusting the identified network structure.

Figure~\ref{wrongB} shows the reconstruction of a random network with $100$ nodes and probability of a connection of $0.2$ for the true value of $\beta=0.03$, and also for various values of $\beta$ deviating up to $\delta=400\%$ from the true value of $\beta$.
The cut-off for the truncated singular value decomposition method is $0.55$ and the probability of \textit{type I error} is $\alpha=0.05$. 
Also in this case, the method is robust to large perturbations of $\beta$, in both negative and positive directions.
A deviation of more than $400\%$ is needed for the method to fail and not to have an improvement over the na\"ive approach.

Figure~\ref{wrongA} shows the reconstruction of a random network with $100$ nodes and probability of a connection of $0.2$ for the true value of $\alpha=0.05$, and $\alpha$ deviating up to $-85\%$.
The cut-off for the truncated singular value decomposition method is $0.6$ and the probability of \textit{type II error} is $\beta=0.03$. 
Figure~\ref{wrongA} shows that the method is affected by relatively large perturbations of $\alpha$.
Namely, for $\delta<-85\%$ and $\delta>50\%$, the reconstructions deviate significantly from the true one.
The reason is that sparse networks are susceptible to perturbation of \textit{type I error}.
Figure~\ref{wrongA2} shows the reconstruction of a denser network, i.e., a random network with probability of a connection of $0.8$, for the same true values of $\alpha$ and $\beta$.
In this case, a deviation of $150\%$ or more is needed for the method to fail.
Comparing Figs.~\ref{wrongA} and \ref{wrongA2}, we can conclude that dense networks are more robust to perturbations of \textit{type I error} than sparse networks.
This is intuitively motivated by the fact that the \textit{type I error} affects links that are not present in the network, and therefore it has a bigger influence on a sparse network.

The above results demonstrated that a rough estimate for $\alpha$ and $\beta$ is sufficient to get an accurate reconstruction; the method is robust to relatively large perturbations of these two errors. 
Rough estimates of these parameters are typically available from simulation studies or prior knowledge about the system. 
Note again that the role of $\alpha$ and $\beta$ are different; $\alpha$ is often controlled and can be obtained from known statistics of the techniques under the null hypothesis; $\beta$ is more difficult as the true alternative would need to be known. 
Given the above simulations, our algorithm is more robust with respect to $\beta$ than $\alpha$, which aligns with the different role of these two errors.
As mentioned at the beginning of Section~\ref{results}, we vary $\alpha$ and $\beta$ in the range $1\%-10\%$ mimicking a typical analysis method that has high sensitivity and high specificity. 
Choosing either lower or higher values for the true $\alpha$ and $\beta$ does not affect the general qualitatively result of the analysis, but it changes the range of perturbation that leads to the failure of the reconstruction method.

\begin{figure}[!t]
\includegraphics[width=0.5\textwidth,trim=0.9cm 0cm 1cm 0cm,clip=true]{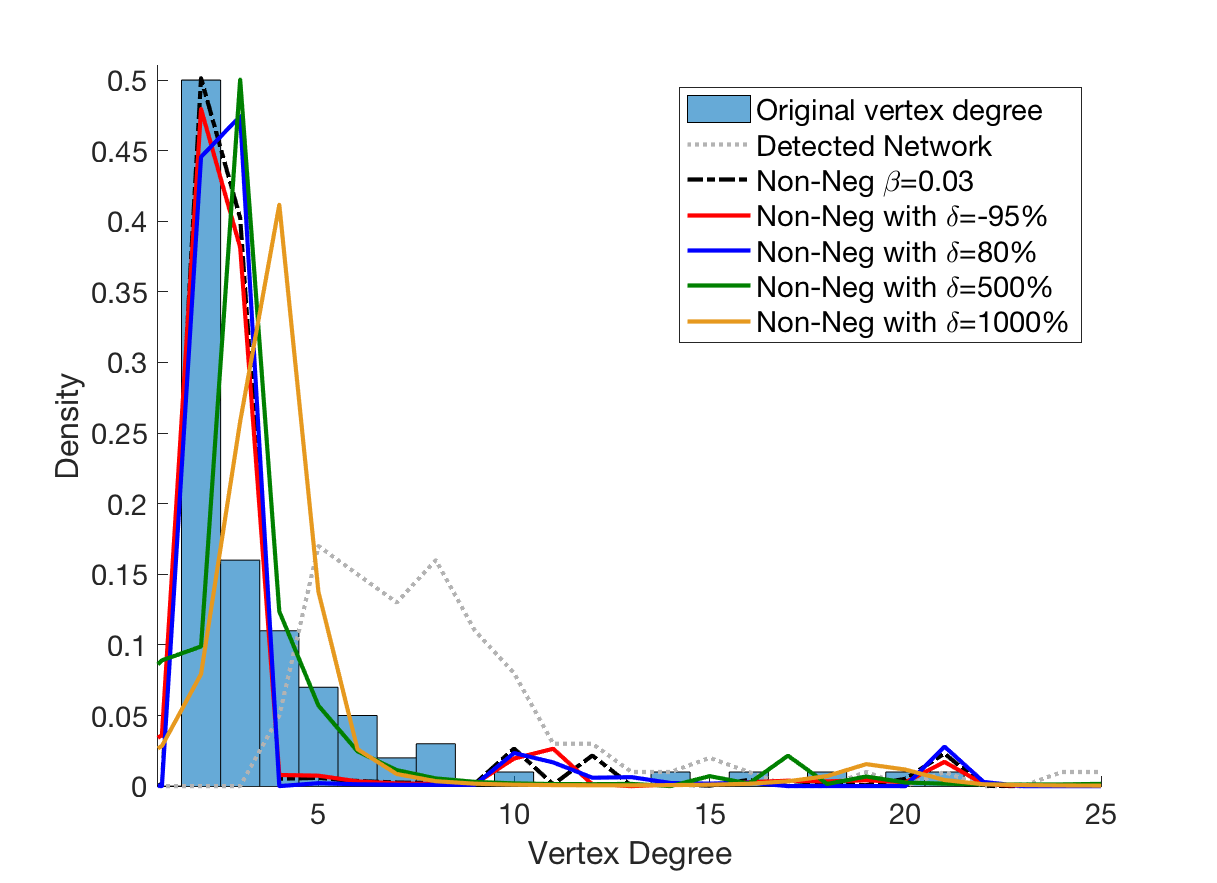}
\caption{Density histogram of the original vertex degrees, blue bars, detected density vertex degree distribution, gray dotted line, and the results when a non-negative constraint is applied to the truncated singular value decomposition using the true $\beta$, black dashed line, and perturbations from the true $\beta$, red, blue, yellow, and green solid lines. The original network is a Scale-Free network with $100$ nodes.}
\label{wrongB2}
\end{figure}

\begin{figure}[!ht]
\includegraphics[width=0.5\textwidth,trim=0.9cm 0cm 1cm 0cm,clip=true]{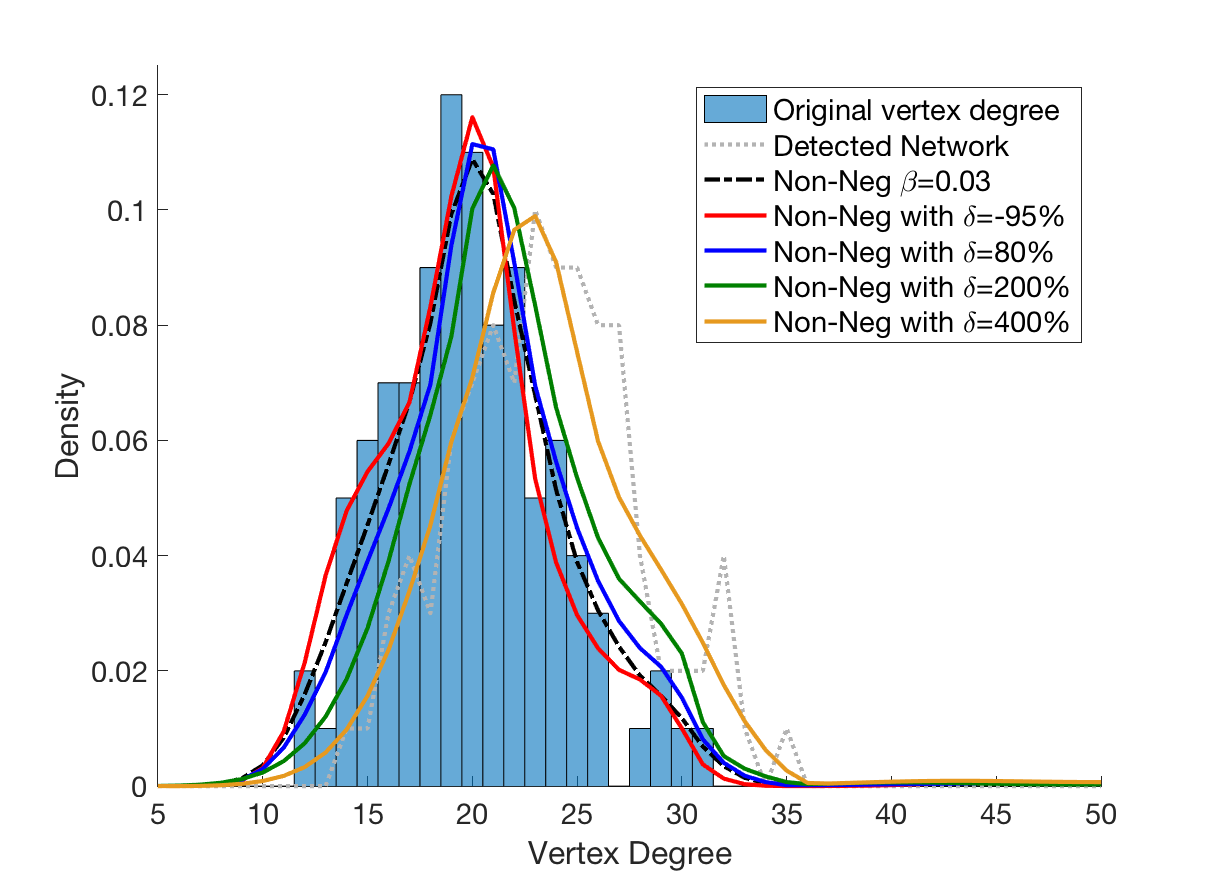}
\caption{Density histogram of the original vertex degrees, blue bars, detected density vertex degree distribution, gray dotted line, and the results when a non-negative constraint is applied to the truncated singular value decomposition using the true $\beta$, black dashed line, and perturbations from the true $\beta$, red, blue, yellow, and green solid lines. The original network is a random network with $100$ nodes and probability of connection $0.2$.}
\label{wrongB}
\end{figure}

\begin{figure}[!ht]
\includegraphics[width=0.5\textwidth,trim=0.9cm 0cm 1cm 0cm,clip=true]{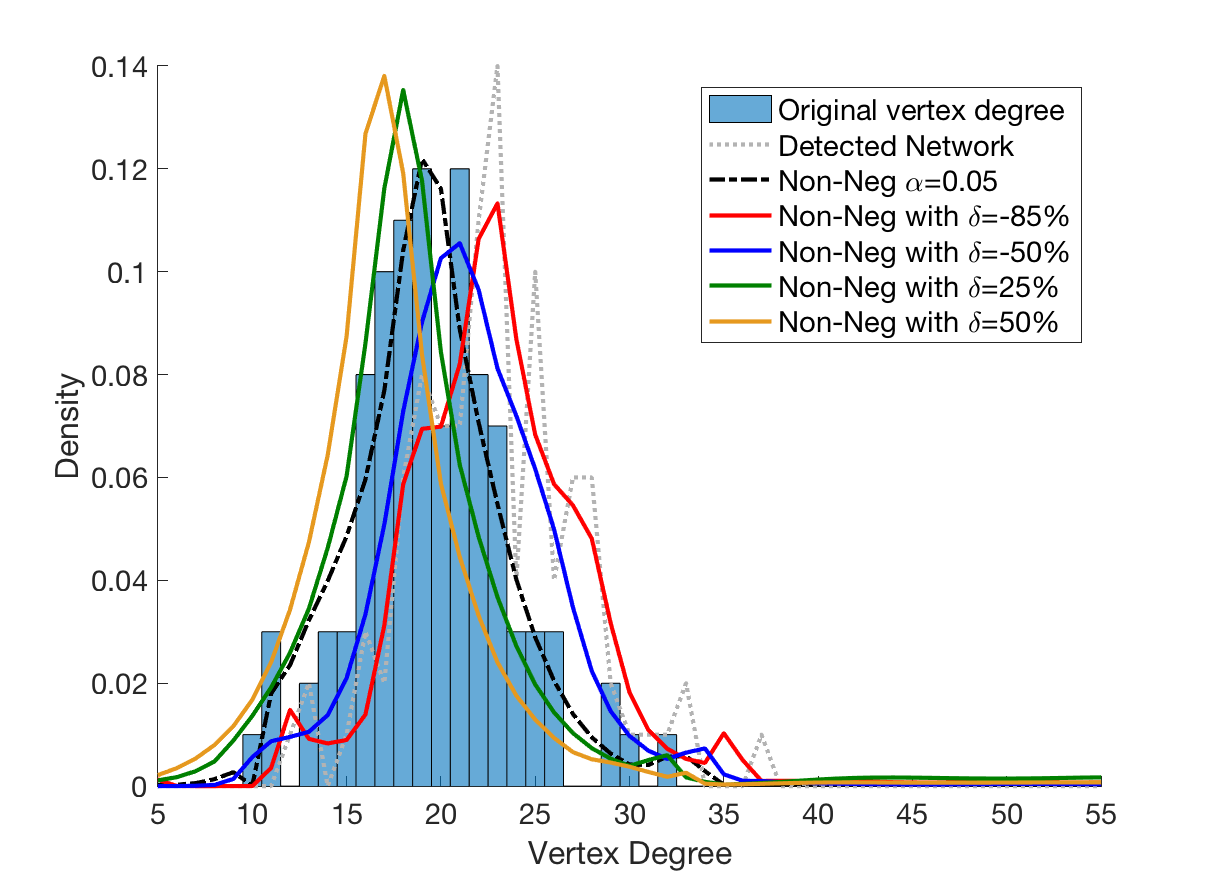}
\caption{Density histogram of the original vertex degrees, blue bars, detected density vertex degree distribution, gray dotted line, and the results when a non-negative constraint is applied to the truncated singular value decomposition using the true $\alpha$, black dashed line, and perturbations from the true $\alpha$, red, blue, yellow, and green solid lines. The original network is a random network with $100$ nodes and probability of connection $0.2$.}
\label{wrongA}
\end{figure}

\begin{figure}[!ht]
\includegraphics[width=0.5\textwidth,trim=0.9cm 0cm 1cm 0cm,clip=true]{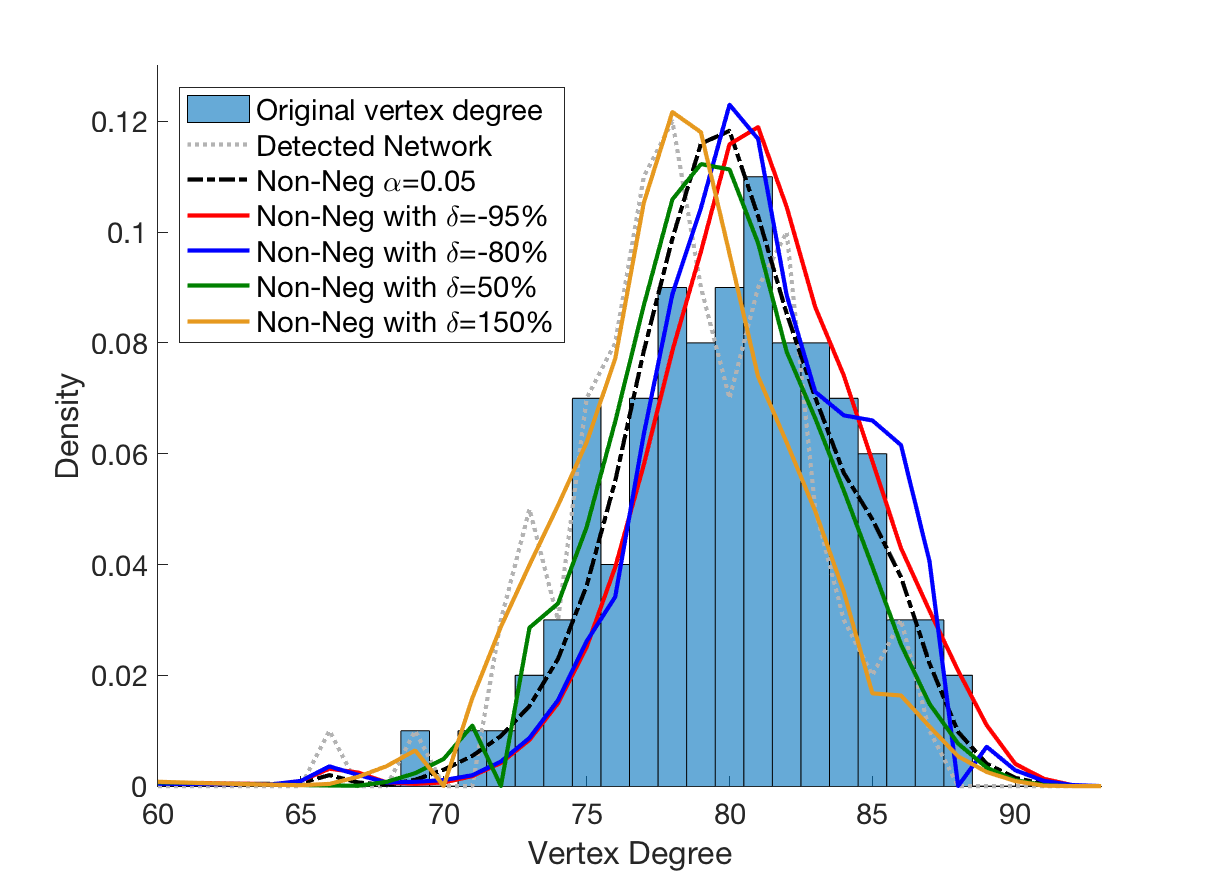}
\caption{Density histogram of the original vertex degrees, blue bars, detected density vertex degree distribution, gray dotted line, and the results when a non-negative constraint is applied to the truncated singular value decomposition using the true $\alpha$, black dashed line, and perturbations from the true $\alpha$, red, blue, yellow, and green solid lines. The original network is a random network with $100$ nodes and probability of connection $0.8$.}
\label{wrongA2}
\end{figure}

\section{Conclusions}

We explore the impact of false positive and false negative conclusions about the presence or absence of links on the vertex degree distribution of a network.
Using an analytical approach, we investigate this dependence on the dimension of the network and the probabilities of \textit{type I} and \textit{type II errors}.
Equation (\ref{p1Ap}) describes the density of the vertex degree distribution of the biased network and thus allows to calculate the influence of false positive and false negative conclusions about links on any kind of network, assuming the probabilities of \textit{type I} and \textit{type II errors} are known.

In the inverse problem, the aim is to reconstruct the original network.
Equation~(\ref{pA+p1}) enables us to calculate analytically the vertex degree distribution of the original network if the biased one and the probabilities of \textit{type I} and \textit{type II errors} are given.
When the dimension of the network is relatively large, numerical issues arise and consequently the truncated singular value decomposition is used to calculate the original network vertex degree distribution.
Numerical simulations show that the vertex degree distribution is correctly recovered in all the cases discussed; the cases presented are designed to cover a variety of network topologies and therefore degree distributions.

The outcomes of this manuscript are general results that enable to reconstruct analytically the vertex degree distribution of any network.
The analytic formula [Eq.~(\ref{pA+p1})] that allows to find the original vertex degree distribution depends only on the detected vertex degree distribution and on the probabilities of \textit{type I} and \textit{type II errors}.
This method is a powerful tool since the vertex degree distribution is a key characteristic of networks.
Moreover, we have actually shown that this method can be used to reconstruct individual node degrees to a very high accuracy. 
This should positively impact on various measures that can be derived from the networks. 
Our proposed method should outperform standard approaches in terms of betweenness centrality, identification of hubs, and other network characteristics. 
This should be rigorously assessed in future research.

A limitation of this work is the assumption that the probabilities of \textit{type I} and \textit{type II errors} are known a priori.
Nevertheless, we show that the method is robust to relatively large perturbations of these two errors.
Therefore, wrong estimates of \textit{type I} and \textit{type II errors}, within certain bounds, do not cause the reconstruction of be rendered invalid. 
We like to emphasise again that in application the \textit{type I error} is typically controlled, while an estimate for the \textit{type II error} can only be obtained through prior experiments/knowledge or simulation studies. 
As shown in various simulations, our reconstruction method is robust to considerable deviations in $\beta$, which supports the usefulness of our technique over and above providing deeper insights into the role of these errors in network reconstruction; our approach is promising for real-world applications. 
Note though that it is always advisable to utilise simulation studies to characterise the advantageous and limitations in a concrete application at hand.
Further analyses should study possible statistical approaches to infer these parameters employing Bayesian approaches or simulation studies. 
We recommend performing the latter to get an estimate of the \textit{type II error} in particular.

Future studies should investigate the influence of \textit{type I} and \textit{type II errors} on other network characteristics, e.g. the number of edges, the global clustering coefficient, and the efficiency.
As a consequence, more information about the original network can be found and, therefore, combining them all a better reconstruction of the network can be achieved.

\begin{acknowledgements}

The authors thank Dr. Daniel Vogel for helpful comments and discussions.
This project has received funding from the European Union's Horizon 2020 research and innovation programme under the Marie Sklodowska-Curie grant agreement No 642563.
The authors declare no competing financial interests.
\end{acknowledgements}

\vspace*{1cm}
\appendix

\section{Determinant of the matrix A}
\label{apProof}

In this appendix we prove that the matrix $A$ has determinant 
$$\det A=(1-\alpha-\beta)^{\frac{n(n-1)}{2}}.$$
To achieve this we have to prove some intermediate steps.

First, we write Eq.~(\ref{bigEq}) in a more compact form as

\begin{widetext}
\begin{equation}
\mathbb{P}(d'=k'|d=k)={\displaystyle \sum_{i=\max\{0,k-k'\}}^{\min\{k,n-1-k'\}}\binom{k}{i}(1-\beta)^{k-i}\beta^i\binom{n-1-k}{k'-k+i} \alpha^{k'-k+i}(1-\alpha)^{n-1-k'-i}}.
\label{compactEq}
\end{equation}
\end{widetext}

Since the element $A_{uv}$ is defined as the probability $\mathbb{P}(d'=u+1|d=v+1)$, we can write 
\begin{widetext}
\begin{equation}
A_{uv}={\displaystyle \sum_{i=\max\{0,v-u\}}^{\min\{v-1,n-u\}}\binom{v-1}{i}(1-\beta)^{v-1-i}\beta^i\binom{n-v}{u-v+i} \alpha^{u-v+i}(1-\alpha)^{n-u-i}},
\label{defA}
\end{equation}
\end{widetext}
for $u,v\in\{1,\cdots,n\}$ and real numbers $0\leq\alpha,\beta\leq 1$.

\begin{prop}[Limit cases]
\label{propLimit}
Let $A=A(n,\alpha,\beta)$ be the matrix defined by Eq.~(\ref{defA}).
For $\beta=1-\alpha$ or $\alpha,\beta=0,1$, the determinant of $A$ satisfies Eq.~\ref{detA}.
\end{prop}
\begin{proof}
When $\beta=1-\alpha$, the Eq.~(\ref{defA}) becomes
$$A_{uv}(n,\alpha,1-\alpha)=\binom{n-1}{u-1}\alpha^{u-1}(1-\alpha)^{n-u}.$$
Note that $A_{uv}(n,\alpha,1-\alpha)$ does not depend on $v$ but only on $u$, therefore in each line all the elements are identical, i.e., it is a multiple of vector $[1,\cdots,1]$; hence, all the lines are linear depend and then the determinant of $A$ is $\det A(n,\alpha,1-\alpha)=0$.

If $\alpha,\beta=0$ then the matrix $A$ is the identity and therefore the determinant is $\det A(n,0,0)=1$.
While if $\alpha,\beta=1$ then the matrix $A$ is anti-diagonal with all elements equal to one, then the determinant is $\det A(n,1,1)=(-1)^{\frac{n(n-1)}{2}}$.

If $\alpha =0,\beta\neq 0$, Eq.~(\ref{defA}) can be formulated as
\begin{equation}
A_{uv}(n,0,\beta)=
\begin{cases}
\binom{v-1}{u-1}(1-\beta)^{u-1}\beta^{v-u} & u\leq v \\
0 & u>v
\end{cases}
\label{alpha0}
\end{equation}
Note that $\beta\neq 1$ since the case $\beta = 1-\alpha$ has been already considered.
The matrix $A(n,0,\beta)$ is upper triangular and therefore the determinant is the product of the elements on the diagonal $A_{uu}(n,0,\beta)=(1-\beta)^{u-1}$, i.e.,
$$\det A(n,0,\beta)=(1-\beta)^{\frac{n(n-1)}{2}}.$$

If $\alpha =1,\beta\neq 0,1$, Eq.~(\ref{defA}) can be formulated as
\begin{equation}
A_{uv}(n,1,\beta)=
\begin{cases}
\binom{v-1}{n-u}(1-\beta)^{v-1-n+u}\beta^{n-u} & u\geq n-v+1\\
0 & u< n-v+1
\end{cases}
\label{alpha1}
\end{equation}
The matrix $A(n,1,\beta)$ has all zeros above the anti-diagonal and therefore the determinant is the product of the elements on the anti-diagonal $A_{uu}(n,1,\beta)=\beta^{u-1}$ and sign given by $(-1)^{\frac{n(n-1)}{2}}$, i.e.,
$$\det A(n,1,\beta)=(-\beta)^{\frac{n(n-1)}{2}}.$$

If $\beta =0,\alpha\neq 0,1$, Eq.~(\ref{defA}) can be formulated as
\begin{equation}
A_{uv}(n,\alpha,0)=
\begin{cases}
\binom{n-v}{u-v}\alpha^{u-v}(1-\alpha)^{n-u} & u\geq v \\
0 & u<v
\end{cases}
\label{beta0}
\end{equation}
The matrix $A(n,\alpha,0)$ is lower triangular and therefore the determinant is the product of the elements on the diagonal $A_{uu}(n,\alpha,0)=(1-\alpha)^{n-u}$, i.e.,
$$\det A(n,\alpha,0)=(1-\alpha)^{\frac{n(n-1)}{2}}.$$

If $\beta =1,\alpha\neq 0,1$, Eq.~(\ref{defA}) can be formulated as
\begin{equation}
A_{uv}(n,\alpha,1)=
\begin{cases}
\binom{n-v}{u}\alpha^{u}(1-\alpha)^{n-u-v+1} & u\leq n-v+1 \\
0 & u>n-v+1
\end{cases}
\label{beta1}
\end{equation}
The matrix $A(n,\alpha,1)$ has all zeros below the anti-diagonal and therefore the determinant is the product of the elements on the anti-diagonal $A_{uu}(n,\alpha,0)=\alpha^{n-u}$ and sign given by $(-1)^{\frac{n(n-1)}{2}}$, i.e.,
$$\det A(n,\alpha,1)=(-\alpha)^{\frac{n(n-1)}{2}}.$$
\end{proof}

Future calculations result easier if the transpose $A^T$ of matrix $A$ is considered.
Considering $A^T$ instead of $A$ does not affect the calculation of the determinant since it is in general true that $\det A^T=\det A$.

\begin{prop}[Transformations]
\label{theoinduction}
Given the matrix A, defined by Eq.~(\ref{defA}), let's call $A^{T_n}$ the transpose of $A$ of dimension $n$.
Let be $0<\alpha,\beta<1$ and $\beta\neq 1-\alpha$.
We call $\overline{A^{T_n}}$ the matrix with elements
\begin{equation}
\overline{a^n_{ij}}=
\begin{cases}
{\displaystyle a^n_{ij}\frac{(1-\alpha-\beta)^{n-1}}{(1-\alpha)^{n-1}}} &\qquad i=1\\
{\displaystyle \left( a^n_{ij}-\frac{a^n_{i1}a^n_{1j}}{a^n_{11}}\right)\frac{1-\alpha}{1-\alpha-\beta}} &\qquad i=2\\
{\displaystyle \left( a^n_{ij}-\frac{\beta}{1-\alpha}a^n_{i-1,j}\right)\frac{1-\alpha}{1-\alpha-\beta}} &\qquad i=3,\cdots,n\\
\end{cases}
\label{transf}
\end{equation}
where $a^n_{ij}$ are the elements of the matrix $A^{T_n}$.
Then, we prove that
\begin{equation}
\overline{A^{T_n}}=
\left[
\begin{array}{c|c}
 (1-\alpha +\beta)^{n-1} &  \\
 \hline
 &\\
 0 & A^{T_{n-1}}  \\
 &\\
\end{array}
\right].
\label{induction}
\end{equation}
\end{prop}

\begin{proof}
To verify Eq.~(\ref{induction}), we have proved that the identity $\overline{a^n_{i+1,j+1}}=a^{n-1}_{ij}$ i.e.,
$$ \left( a^n_{2j}-\frac{a^n_{21}a^n_{1j}}{a^n_{11}}\right)\frac{1-\alpha}{1-\alpha-\beta}= a^{n-1}_{1j}$$
and
\begin{multline*}
\left( a^n_{i+1,j+1}-\frac{\beta}{1-\alpha}a^n_{i,j+1}\right)\frac{1-\alpha}{1-\alpha-\beta}=a^{n-1}_{ij} \\ \text{for}\ i=2,\cdots n-1
\end{multline*}
are always valid.
To achieve this, we have split the calculations into cases, i.e.,
\begin{equation*}
i=j,\quad
\begin{cases}
j>i\\
j\geq n-i-1
\end{cases},
\begin{cases}
j>i\\
j<n-i-1
\end{cases},
\end{equation*}
\begin{equation*}
\begin{cases}
j<i\\
j>n-i
\end{cases},
\begin{cases}
j<i\\
j<n-i
\end{cases},
\begin{cases}
j<i\\
j=n-i
\end{cases}.
\end{equation*}
For each condition the identity $\overline{a^n_{i+1,j+1}}=a^{n-1}_{ij}$ has been proved.
\end{proof}

\begin{theorem}
\label{theoDet}
The $n\times n$ matrix A, defined by Eq.~(\ref{defA}), has determinant 
\begin{equation*}
\det A=(1-\alpha-\beta)^{\frac{n(n-1)}{2}}.
\end{equation*}
\end{theorem}

\begin{proof}
For $\alpha,\beta=0,1$ or $\beta=1-\alpha$, Proposition~\ref{propLimit} proves the theorem.
Assume $0<\alpha,\beta<1$ and $\beta\neq 1-\alpha$.
Since a matrix and its transposed have the same determinant, we proceed considering the matrix $A^{T_n}$ proving a proof by induction.

The base of induction is $n=2$; in this case the matrix is
\begin{equation*}
A^{T_2}=
\left[
\begin{array}{ccc}
1-\alpha  & \alpha  \\
\beta  & 1-\beta\\
\end{array}
\right]
\end{equation*}
and it has determinant $\det A^{T_2}=1-\alpha-\beta$.

The inductive step consists in assuming that $\det A^{T_{n-1}}=(1-\alpha-\beta)^{\frac{(n-1)(n-2)}{2}}$, i.e., the inductive hypothesis for dimension $n-1$, and proving the statement for dimension $n$.

Since $0<\alpha,\beta<1$ and $\beta\neq 1-\alpha$, we can apply Proposition~\ref{theoinduction}.
The transformations defined by Eq.~(\ref{transf}) guarantees that the matrix $\overline{A^{T_n}}$ has the same determinant as $A^{T_n}$, and according to Eq.~(\ref{induction}), we can express the determinant as $\det A^{T_{n}}=(1-\alpha-\beta)^{n-1}\det A^{T_{n-1}}$.
We can now apply the inductive hypothesis, therefore

\begin{align*}
\det A =& \det A^{T_n} \\
=& \det \overline{A^{T_n}}\\
=& (1-\alpha-\beta)^{n-1}\det A^{T_{n-1}}\\
=& (1-\alpha-\beta)^{n-1} (1-\alpha-\beta)^{\frac{(n-1)(n-2)}{2}}\\
=& (1-\alpha-\beta)^{\frac{n(n-1)}{2}}
\end{align*}
that concludes the proof of the theorem.
\end{proof}

\bibliographystyle{apsrev4-1}
\bibliography{prE.bib}

\end{document}